\documentclass[sigconf,natbib=false,nonacm]{acmart}

\usepackage{graphicx}

\usepackage{amsmath}
\usepackage{mathtools}
\usepackage{thmtools}
\usepackage{amsfonts}
\usepackage{algorithm}
\usepackage{xspace}
\usepackage[capitalize]{cleveref}
\usepackage[noend]{algpseudocode}
\usepackage{tikz}
\usetikzlibrary{arrows.meta}

\usepackage[
    doi         = true,
    giveninits  = true,
    isbn        = false,
    sortcites   = true,
    style       = acmnumeric,
    url         = false,
    maxnames = 99,
    minnames = 99,
]{biblatex}

\AtEveryBibitem{
    \clearfield{month}
    \clearlist{location}
    \clearlist{publisher}
    \clearname{editor}
}

\makeatletter
\newcommand{\DeclareMathActive}[2]{%
  \expandafter\edef\csname keep@#1@code\endcsname{\mathchar\the\mathcode`#1 }
  \begingroup\lccode`~=`#1\relax
  \lowercase{\endgroup\def~}{#2}%
  \AtBeginDocument{\mathcode`#1="8000}%
}

\newcommand{\std}[1]{\csname keep@#1@code\endcsname}
\patchcmd{\newmcodes@}{\mathcode`\-\relax}{\std@minuscode\relax}{}{\ddt}
\AtBeginDocument{\edef\std@minuscode{\the\mathcode`-}}
\makeatother

\DeclareMathActive{O}{\operatorname{\std{O}}}

\algnewcommand{\AND}{\textbf{ and }}
\algnewcommand{\OR}{\textbf{ or }}
\graphicspath{ {./figures/} }

\algrenewcommand\alglinenumber[1]{\footnotesize #1}

\def\Bin{\operatorname{Bin}}

\def\I_#1{\ensuremath{I_{\text{#1}}}}
\def\Time#1{\ensuremath{#1.\mathsf{time}}}
\def\Max#1{\ensuremath{#1.\mathsf{max}}}
\def\LastMax#1{\ensuremath{#1.\mathsf{lastMax}}}
\def\Interactions#1{\ensuremath{#1.\mathsf{interactions}}}
\def\Geom{\operatorname{Geom}}

\def\PlainTime{\ensuremath{\mathsf{time}}\xspace}
\def\PlainMax{\ensuremath{\mathsf{max}}\xspace}
\def\PlainLastMax{\ensuremath{\mathsf{lastMax}}\xspace}
\def\PlainInteractions{\ensuremath{\mathsf{interactions}}\xspace}

\begin{abstract}
The population protocol model describes collections of distributed agents that interact in pairs to solve a common task.
We consider a \emph{dynamic} variant of this prominent model, where we assume that an adversary may change the population size at an arbitrary point in time.
In this model we tackle the problem of \emph{counting the population size:} in the dynamic size counting problem the goal is to design an algorithm that computes an approximation of $\log n$.
This estimate can be used to turn static, non-uniform population protocols, i.e., protocols that depend on the population size $n$, into dynamic and loosely-stabilizing protocols.

Our contributions in this paper are three-fold.
Starting from an arbitrary initial configuration, we first prove that the agents converge quickly to a valid configuration where each agent has a constant-factor approximation of $\log n$, and once the agents reach such a valid configuration, they stay in it for a polynomial number of time steps.
Second, we show how to use our protocol to define a uniform and loosely-stabilizing phase clock for the population protocol model.
Finally, we support our theoretical findings by empirical simulations that show that our protocols work well in practice.
\end{abstract}



\copyrightyear{2024}
\acmYear{2024}
\setcopyright{rightsretained}
\acmConference[PODC '24]{ACM Symposium on Principles of Distributed Computing}{June 17--21, 2024}{Nantes, France}
\acmBooktitle{ACM Symposium on Principles of Distributed Computing (PODC '24), June 17--21, 2024, Nantes, France}\acmDOI{10.1145/3662158.3662825}
\acmISBN{979-8-4007-0668-4/24/06}

\begin{document}

\title{Dynamic Size Counting in the Population Protocol Model}
\subtitle{Counting (on) agents to drive a phase clock}

\author{Dominik Kaaser}
\email{dominik.kaaser@tuhh.de}
\orcid{0000-0002-2083-7145}

\affiliation{%
  \institution{TU Hamburg}
  \city{Hamburg}
  \country{Germany}
}

\author{Maximilian Lohmann}
\email{maximilian.lohmann@tuhh.de}
\orcid{0009-0007-1791-7906}

\affiliation{%
  \institution{Christian Doppler Laboratory for Blockchain Technologies for the Internet of Things, TU Hamburg}
  \city{Hamburg}
  \country{Germany}
}

\maketitle

\vfill
\mbox{}

\section{Introduction}
\label{sec:introduction}

Population protocols \cite{DBLP:journals/dc/AngluinADFP06} are a prominent model for computation among simple and anonymous distributed \emph{agents}.
Agents have limited memory and are commonly modeled as finite state machines.
They interact with each other in pairs where they observe each others' states and update their states according to a common transition function.
In each time step a pair of agents is chosen uniformly at random to interact.
The computation does not halt but the agents rather converge to a desired configuration or output.
Population protocols have a multitude of applications.
They are closely related to chemical reaction networks \cite{DBLP:journals/nc/SoloveichikCWB08} and model interacting particle systems \cite{liggett1985interacting}, programmable
chemical controllers at the level of DNA \cite{chen2013}, or biochemical regulatory processes in living cells \cite{cardelli2012}.

There are two prohibitive factors that limit the application of classical population protocols in realistic scenarios:
all agents start in a predefined state, and the population size $n$ is fixed over time.
At the same time, in their original paper \textcite{DBLP:journals/dc/AngluinADFP06} motivate the model with a flock of birds that are equipped with temperature sensors.
Clearly, the number of birds in a flock changes over time \cite{DBLP:conf/dcoss/Delporte-GalletFGR06}.
Even worse, throughout hunting season there is a looming threat that a poaching adversary selectively targets certain types of birds in the flock.
In this paper we therefore design a \emph{uniform dynamic size counting} protocol that can be used to transform modern (non-uniform) protocols into dynamic protocols.
In a uniform population protocol the state space and the transition function do not depend on the population size $n$.

We prove that our protocol is \emph{loosely-stabilizing} according to the definition by \textcite{DBLP:journals/tcs/SudoNYOKM12} even if the population size changes over time: when started in any arbitrary state, our protocol quickly reaches a state with correct output and remains in such a state for a polynomial number of interactions.
Our protocol improves upon the state space complexity of the best known protocol by \textcite{DBLP:journals/tcs/DotyE21} at the expense of a slightly increased running time.

A plenitude of recent works (e.g., \cite{DBLP:conf/soda/AlistarhAEGR17, DBLP:conf/soda/AlistarhAG18,DBLP:conf/podc/BilkeCER17,DBLP:conf/wdag/BerenbrinkEFKKR18, DBLP:conf/focs/DotyEGSUS21,DBLP:conf/focs/DotyEGSUS21,DBLP:conf/stoc/BerenbrinkGK20})
have proposed efficient but non-uniform population protocols where the required number of states slowly grows with the population size, and the transition functions encode a function of $n$.
Unfortunately, in the context of biologically-inspired computing, determining the exact population size or even a polynomial approximation may not be feasible.
Size counting protocols \cite{DBLP:conf/podc/BerenbrinkKR19,DBLP:conf/podc/DotyE19,DBLP:journals/tcs/DotyE21} address precisely this issue: they provide an approximation of the number of agents which can then be used to execute non-uniform payload protocols.

\subsection{Results in a Nutshell}
The problem of dynamic size counting is equivalent to the problem of loosely-stabilizing size counting \cite{DBLP:conf/sand/DotyE22}.
In this paper we therefore present the first protocol using an asymptotically optimal number of memory bits for the loosely-stabilizing size counting problem.
Our size counting protocol doubles as a loosely-stabilizing uniform phase clock that can be used to synchronize populations.

We remark that the clock by \textcite{DBLP:conf/sand/BerenbrinkBHK22} is also loosely-stabilizing.
However, their clock encodes $\log n$ in its transition function and thus is not uniform.
Clearly, once a protocol encodes the population size it cannot be applied in our dynamic setting where the population size may change over time.

Assuming $\log \hat{n}$ to be the largest initial estimate among the population, our protocol converges in $O(\log \hat{n} + \log n)$ parallel time w.h.p.
All agents then continue holding correct estimates for polynomial parallel time w.h.p.
When the largest value initially stored in any agents' state is $s$, our protocol requires $O(\log s + \log \log n)$ bits of memory w.h.p.
The expression \emph{with high probability (w.h.p.)} denotes a probability of at least $1 - n^{-\Omega(1)}$.

\subsection{Background and Related Work}

\paragraph{Approximate Counting.}
\textcite{DBLP:conf/soda/AlistarhAEGR17} introduce a protocol for approximate size counting using coin flips.
Each agent flips a coin until it lands on heads.
The resulting number of trials has (almost) a geometric distribution with parameter $1/2$.
The largest value sampled by one of the $n$ agents is a constant-factor approximation of $\log n$.
In the original population protocol model, however, agents do not have access to a source of randomness.
Thus, so-called synthetic coins are used \cite{DBLP:conf/soda/AlistarhAEGR17,DBLP:conf/soda/BerenbrinkKKO18}.
This technique does not require an external source of randomness but instead it extracts randomness from the random scheduler.

To simplify the analysis of size counting protocols using random coins, \textcite{DBLP:conf/podc/DotyE19} assume that agents have access to random bits.
They present an approximate size counting protocol that computes an estimate of $\log n \pm 5.7$.
This is accomplished by taking the average of $O(\log n)$ maxima of $n$ geometrically distributed random variables each.

\Textcite{DBLP:conf/podc/SudoOIKM19} define a protocol that provides truly uniform synthetic random coins by splitting the population into two groups, \emph{leaders} and \emph{followers}. Leaders generate random coins by observing whether they are initiator or responder in an interaction.
Unfortunately, this approach cannot be used in our dynamic setting since it may occur that all leaders are removed from the population.

\Textcite{DBLP:conf/podc/BerenbrinkKR19} present an approximate size counting protocol that computes $\lfloor \log n \rfloor$ or $\lceil \log n \rceil$.
Here, the agents elect a leader that generates $M$ tokens which are balanced using a load-balancing algorithm.
If some agents do not have a token after balancing, $M$ must have been smaller than $n$.
In this case, $M$ is doubled and the load balancing restarts.
When the process terminates, the agents learn $\log n \pm 1$ from the leader.
As before, this protocol is also not suitable in the dynamic setting where the single leader agent may be removed from the population.

For further details we refer to the survey of size counting in population protocols by \textcite{DBLP:journals/tcs/DotyE21}.
There, the authors also introduce a mechanism to compose uniform counting protocols with non-uniform protocols.
Nevertheless, their compound protocols are not applicable in the dynamic setting where the population size may change over time.

\paragraph{Dynamic Size Counting}
Dynamic population protocols are one of the open problems described in the seminal paper by \textcite{DBLP:journals/dc/AngluinADFP06}.
In this setting \textcite{DBLP:conf/sand/DotyE22} introduce an \emph{adversary} that is capable of adding and removing agents from the population.
All agents are added in some predefined state, but agents can be removed arbitrarily.
With regards to counting the authors prove that dynamic size counting is equivalent to loosely-stabilizing size counting, and they present a loosely-stabilizing size counting protocol that solves the dynamic size counting problem.

In their protocol, \textcite{DBLP:conf/sand/DotyE22} use successive random coin flips to approximate $\log n$ w.h.p. 
However, the naive approach of always spreading the largest estimate breaks as soon as the population shrinks.
Instead, they use the robust detection algorithm introduced in \cite{DBLP:conf/dna/AlistarhDKSU17} to detect when the estimate is inaccurate.
To this end, each agent stores a list of length $O(\log n)$ which can be reduced to $O(\log \log n)$ at the expense of a reduced sensitivity to population changes.
The idea is to not only use the maximum of $n$ geometrically distributed random variables but to use the detection algorithm from \cite{DBLP:conf/dna/AlistarhDKSU17} on their first missing value.
As a result, every agent uses $O(\log^2 s + \log n \log \log n)$ bits or ${O(\log^2 s + (\log \log n)^2)}$ bits of memory w.h.p.\ in the optimized version.
Here, $s$ describes the maximum value stored in any variable of the agent's memory.
For this protocol \citeauthor{DBLP:conf/sand/DotyE22} prove loose-stabilization:
starting from an arbitrary configuration their protocol converges in ${O(\log n + \log \log \hat{n})}$ time, and once the protocol has converged all agents hold an estimate of $\Theta(\log n)$ for a polynomial time.
Here, $\log \hat{n}$ is the initial maximum estimate among the population.

\paragraph{Phase Clocks.}
Phase clocks have widespread use in distributed systems.
They can be used for mode changes, triggering different behavior, periodic resets to a predefined state, or periodic synchronization of a \emph{time} variable  \cite{DBLP:journals/ppl/AroraDG91}.
In population protocols, phase clocks can be categorized into leader driven, junta driven (i.e., by a group of leaders), and leaderless.
An agent's \emph{time} can be visualized as a hand on a clock face.
All agents' hands should roughly point in the same direction, i.e., lie within a specific interval.
The phase clock is divided into hours, each requiring $\Theta(\log n)$ time to pass.

In the context of population protocols, phase clocks are a fundamental building block many efficient population protocols rely on to orchestrate their distributed computation.
Intuitively, phase clocks divide time into blocks of $\Theta(n \log n)$ interactions each.
These blocks are then sufficiently long such that a simple epidemic spreading process succeeds in transmitting some information to all agents.
This allows block-synchronization of the population and, ultimately, efficient algorithms for the population protocol model.

The first phase clocks were introduced and analyzed by \textcite{DBLP:journals/dc/AngluinAE08a}. 
Together with a \emph{junta-election} mechanism~\cite{DBLP:journals/jacm/GasieniecS21} they have been employed to solve leader election \cite{DBLP:conf/soda/GasieniecS18}, majority \cite{DBLP:journals/dc/BerenbrinkEFKKR21,DBLP:conf/sand/BerenbrinkBHK22}, plurality consensus \cite{DBLP:conf/podc/BankhamerBBEHKK22,DBLP:conf/soda/BankhamerBBEHKK22}, and size counting \cite{DBLP:conf/podc/BerenbrinkKR19}. 
Simple phase clocks are implemented by counters modulo some large value $m$ (see, e.g., \cite{DBLP:conf/soda/AlistarhAG18,DBLP:journals/dc/AngluinAE08a,DBLP:journals/dc/BerenbrinkEFKKR21,DBLP:journals/jacm/DolevW04,DBLP:conf/soda/GasieniecS18,DBLP:conf/sand/BerenbrinkBHK22}).
Whenever the counter of some agent $u$ crosses zero, the agent receives a \emph{signal} indicating that a new phase  starts.
Following the nomenclature by \textcite{DBLP:conf/sand/BerenbrinkBHK22}, the signals divide time into alternating intervals: so-called \emph{burst-intervals} and so-called \emph{overlap-intervals}.
During a burst-interval, every agent gets exactly one signal.
An overlap-interval consists of those interactions between two burst-intervals where no agent gets a signal.
A burst-interval together with the subsequent overlap-interval forms a \emph{phase}.

\Textcite{DBLP:conf/sand/BerenbrinkBHK22} present a loosely-stabilizing phase clock protocol that quickly synchronizes arbitrary configurations.
This leaderless phase clock requires $O(\log n)$ states and runs forever.
After taking at most $O(\log n)$ time to enter a synchronous configuration, the population stays synchronized at any time w.h.p.
However, their protocol is non-uniform: it requires an approximation of $\log n$ and is thus unsuitable to generate an approximation of $\log n$ in the dynamic setting.
Nevertheless, while their exact approach is not applicable in our setting, their three-phase division of a phase clock inspired our loosely-stabilizing protocol.

\paragraph{Detection.}
\Textcite{DBLP:conf/dna/AlistarhDKSU17} propose a protocol called \emph{detection} that allows all agents to learn whether a so-called \emph{source} agent is present in the population.
The idea is to use transitions of the form ${(u, v) \rightarrow (\min\{u+1, v+1\}, \min\{u+1, v+1\})}$ except for source agents which do not change their state but stay at zero.
If the agents reach some value in $\Omega(\log n)$, no source is present w.h.p.
In 
\cite{DBLP:conf/dna/AlistarhDKSU17} the authors prove lower bounds for the minimum value of any agent if no source agent is present. \Textcite{DBLP:journals/tcs/SudoNYOKM12} prove corresponding upper bounds.
\Textcite{DBLP:conf/wdag/SudoEIM21} analyze a related transition function $(u, v) \rightarrow (\max\{u-1, v-1\}, \max\{u - 1, v - 1\})$ for a protocol called Countdown with Higher Value Propagation (CHVP).
In this paper, we use the CHVP protocol to synchronize the population.

\section{Model and Results}
\label{chapter:2}

Throughout this paper, we use $V$ to denote the set of all agents and $|V| = n$ to denote the number of agents.
The \emph{state space} $Q$ consists of tuples, where each entry is called a \emph{variable}.
A \emph{configuration} $C\colon V \rightarrow Q$ maps each agent to a state in $Q$.
The \emph{execution} of a protocol $\Xi = \{C_0, C_1, \dots \}$ is a sequence of configurations.
The configuration $C_0$ is called the \emph{initial configuration}.
Each configuration $C_{i+1}$ is generated from configuration $C_i$ by selecting two agents uniformly at random and updating their states according to the transition function defined by our algorithm.
Thus $C_i$ describes the configuration after $i$ interactions.
We use pseudocode to describe how two interacting agents change their variables in an interaction.

A protocol is \emph{$(t_c(n), t_h(n))$-loosely-stabilizing} if, from any configuration, it converges in $t_c(n)$ time to a correct configuration and holds a correct configuration for $t_h(n)$ time w.h.p.
Time is measured in parallel time, where one unit of parallel time consists of $n$ interactions.
Accordingly, $C_n$ describes the configuration after one parallel time step.
For the sake of comparability we use the same metric as in \cite{DBLP:conf/sand/DotyE22} and measure space complexity in bits rather than in the number of states.

\subsection{Simplified Algorithm}

\algnewcommand{\LineComment}[1]{\State \(\triangleright\) #1}
\algrenewcommand\algorithmicthen{}
\begin{algorithm}[t]
\caption{SimplifiedDynamicSizeCounting($u, v$)}
\label{alg:dynamic_size_counting_simplified}
\label{alg:simplified}
\begin{algorithmic}[1]
\If{$\Time{u} \leq 0$ \Comment{wrap-around}\\\hspace{.5em}\OR
($u \in \I_{reset} \AND v \in \I_{exchange}$) \Comment{reset $\rightarrow$ exchange} \\\hspace{.5em}\OR
($u \notin \I_{exchange} \AND \Max u \neq \Max v$) \Comment{hold $\rightarrow$ exchange}\\
\textbf{then}
}
    \State $grv \gets \Geom(1/2)$
    \State $(\Time u, \Max u) \gets (\tau_1 \cdot \max\{\Max u, grv\}, grv)$
\EndIf
\vspace{1ex}

\If{$u, v \in \I_{exchange} \AND \Max u < \Max v$} \Comment{exchange maximum}
    \State $(\Time u, \Max u) \gets (\tau_1 \cdot \Max v, \Max v)$
\EndIf

\State $\Time u \gets \max\{\Time u, \Time v\} - 1$ \Comment{update time}

\end{algorithmic}
\end{algorithm}

To give an intuitive overview of our algorithm we first present a simplified version (see \cref{alg:simplified}).
Our protocol uses geometrically distributed random variables (GRVs) to estimate the population size.
These variables can be obtained by repeatedly flipping a coin and counting how many flips are needed until it lands on heads.
(We later explain how agents could generate such random variables from the inherent randomness of the scheduler.) 
The underlying idea is that the maximum of $n$ independent random variables with distribution $\Geom(1/2)$ is in $\Theta(\log n)$ w.h.p.
Thus, the agents generate a linear number of GRVs and spread their maximum via epidemic spreading.
To adapt to changing population sizes, this process is repeated cyclically akin to a countdown timer.

In the simplified case, each agent $u$ stores two variables, \PlainMax and \PlainTime.
The variable \PlainMax stores the current maximum value of GRVs that agent $u$ has encountered.
The variable \PlainTime tracks when the maximum GRV was last adopted by any agent.
It is a simple countdown that starts at a multiple of \PlainMax and uses CHVP to ensure little deviation among the population.
Once the \PlainTime reaches zero, it wraps around to a multiple of either the current maximum value \PlainMax or a newly sampled GRV, if the latter happens to be larger. This wrap-around is called a \emph{reset}.

\begin{figure}[b]
\vspace{-3ex}
\centering
\begin{tikzpicture}[scale=0.33]
\clip (-8,-5) rectangle (8,7);
%\draw (-8,-5) rectangle (8,7);
  \draw[line width=1pt] (0,0) circle (4.5cm); % main circle
  \node (B) at (-1,0) {};
  % Launch phase
  \node (A) at (90:6.0cm) {$0 \rightarrow \tau_1 \cdot \PlainMax$};
  
  \draw [line width=0.5pt,dashed,-{Stealth[scale=1.5]}] (0:4.5cm)  to [out=45,in=20,in looseness=1.5,out looseness=1.25] (90:4.5cm);
  
  % Wait phase
  \draw[line width=1pt] (0,0) -- (-30:4.5cm) arc (-30:90:4.5cm) -- cycle;
  \node[line width=1pt, ->] (B) at (-30:6.5cm) {$\tau_2 \cdot \PlainMax$};
  
  \draw [line width=0.5pt,dashed,-{Stealth[scale=1.5]}] (240:4.5cm)  to [out=160,in=150,in looseness=1.5,out looseness=2.5] (90:4.5cm);
  
  % Gather phase
  \draw[line width=1pt] (0,0) -- (210:4.5cm) arc (210:330:4.5cm) -- cycle;
  \node at (210:6.5cm) {$\tau_3 \cdot \PlainMax$};

  \draw [line width=0.5pt,dashed,-{Stealth[scale=1.5]}] (180:4.5cm)  to [out=135,in=160,in looseness=1.5,out looseness=1.25] (90:4.5cm);
  
  % Labeling
    \node at (30:2.5cm) {Exchange};
    \node at (-90:2.5cm) {Hold};
    \node at (150:2.5cm) {Reset};
\end{tikzpicture}
\caption{Phase transitions}
\Description[The order of phase transitions]{The phases can be displayed as a clock, with the exchange phase being in the right top third, the hold phase in the bottom third, and the reset phase in the top left third.}
\end{figure}
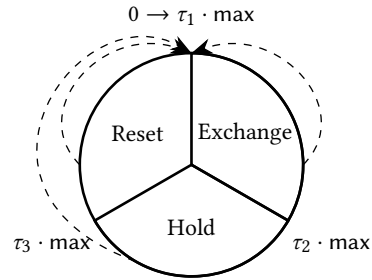

Similarly to \cite{DBLP:conf/sand/BerenbrinkBHK22}, the state space for the \PlainTime variable is divided into three phases which we call \emph{exchange}, \emph{hold}, and \emph{reset}.
The phases are defined based on an agent's current maximum estimate \PlainMax.
To this end we use three constants $\tau_1 > \tau_2 > \tau_3$
and define that an agent $u$ is in the exchange phase if its $\Time u$ variable is in 
$(\tau_1 \cdot \Max u, \tau_2 \cdot \Max u]$, in the hold phase if its $\Time u$ variable is in 
$(\tau_2 \cdot \Max u, \tau_3 \cdot \Max u]$, and in the reset phase if otherwise its $\Time u$ variable is in 
$(\tau_3 \cdot \Max u, 0]$.

In the first phase, all agents exchange their maximum value via epidemic spreading.
When agents adopt a new larger maximum value, their $\PlainTime$ is rewound to $\tau_1 \cdot \PlainMax$.
Agents use the detection protocol~\cite{DBLP:conf/dna/AlistarhDKSU17} to synchronize their $\PlainTime$ by detecting when the maximum was last adopted.
This detection protocol is based on epidemic spreading and it concludes before any agent leaves the exchange phase.
At this point, all agents have the same maximum value w.h.p.

After the exchange phase, the agents move to the hold phase.
If, from this point on, an agent encounters a different maximum value (a low probability event), the premise of the exchange phase has failed.
In this case the agent resets: it generates a new GRV and sets its $\PlainTime$ back to $\tau_1 \cdot \PlainMax$.
Otherwise, if all maximum values align, all agents eventually move to the reset phase.
Here, they will immediately reset when interacting with any agent in the exchange phase.
Thus, when the first agent leaves the reset phase for the exchange phase, all other agents quickly follow.
This effectively starts the next exchange phase via epidemic spreading.

\subsection{Our Results}
While the algorithm presented in \cref{alg:simplified} aims to convey the idea of our protocol, we do need to modify it to facilitate a theoretical analysis.
We will describe the required additions in the next section.
For the amended protocol we show the following theorem.

\begin{theorem}
    \label{theorem:1}
    For any constant $k \geq 2$, and largest initial size estimate $\log \hat{n}$, our algorithm is a $\left(O(\log \hat{n} + \log n), \Theta(n^{k-1} \log n)\right)$-loosely-stabilizing protocol solving the size counting problem w.h.p.
    When $s$ denotes the largest value initially stored in any of the agents' variables, our protocol requires $O(\log s + \log \log n)$ bits per agent w.h.p.
\end{theorem}

We remark that the best known size counting protocols compute $\log n \pm 1$ in the static setting~\cite{DBLP:conf/podc/BerenbrinkKR19}. However, non-uniform protocols typically require only an estimate  $\log\hat n = \Theta(\log n)$ and thus a constant-factor approximation is sufficient for their correctness and performance.

Since our protocol is uniform, it also solves the dynamic size counting problem.
When considering the memory in bits, once our protocol is converged it requires an optimal $O(\log \log n)$ bits per agent, improving upon the $\Omega((\log\log n)^2)$ bits required by \cite{DBLP:conf/sand/DotyE22}.
However, this decrease in space complexity comes at the expense of a larger convergence time when compared to \cite{DBLP:conf/sand/DotyE22}.
In particular, our convergence time depends on $\log\hat n + \log n$ (where $\log \hat n$ is the initial estimate) while their protocol converges in roughly $\log\log\hat n + \log n$ time. This means their protocol is faster when the initial population size is exponentially over-estimated.

\medskip
In addition, our protocol doubles as a loosely-stabilizing uniform phase clock.
We say that an agent receives a signal whenever the agent resets.
Formally, we show the following theorem.

 \begin{theorem}\label{thm:clock}
Let $k \geq 2$ be a constant. Assume that all agents have an estimate of $\Theta(\log n)$ at time $t_0$. Then there exists a constants $c$ and a sequence of time steps $(t_i)_{i \geq n}$ such that every agent ticks exactly once in the interval $[t_i-c\cdot n\log n, t_i + c\cdot n\log n]$ (``burst'') and $t_{i+1} - t_i = \Theta(n \log n)$ with $t_{i+1} - t_i \geq 3c \cdot n\log n$ (``overlap'') for a polynomial number of intervals $i \leq n^{k}$.
 \end{theorem}

\algrenewcommand\algorithmicthen{}

\begin{algorithm*}[t]
\caption{DynamicSizeCounting($u, v$)}
\label{alg:dynamic_size_counting}
\begin{algorithmic}[1]
\Function{DynamicSizeCounting}{$u, v$}
\If{$\Time{u} \leq 0$ \Comment{wrap-around}\label{dyn:phase_trans_cond}\\
\hspace{2em}\OR ($u \in \I_{reset} \AND v \in \I_{exchange}$) \Comment{reset $\rightarrow$ exchange}\\
\hspace{2em}\OR ($u \notin \I_{exchange} \AND \Max u \neq \Max v$) \textbf{then}\Comment{hold $\rightarrow$ exchange}
}
    \State $grv \gets 20(k+1) \cdot GRV(k)$
    \State $(\Time u, \Interactions u, \Max u, \LastMax u) \gets (\tau_1 \cdot \max\{\Max u, grv\}, 0, grv, \Max u)$ \label{dyn:reset_with_sample}
\EndIf

\vspace{1ex}
\If{$\Interactions u > \tau' \cdot \max\{\Max u, \LastMax u\}$} \Comment{``backup'' GRV generation} \label{dyn:backup_grv_start}
    \State $(\Interactions u, grv) \gets (0, GRV(k))$ \label{dyn:backup_grv} 
    \If{$grv > \Max u$} \label{dyn:backup_grv_cond}\Comment{reset if larger than overestimated $\PlainMax$}
        \State $(\Time u, \Max u) \gets (\tau_1 \cdot 20(k+1) \cdot grv, 20(k+1) \cdot grv)$ \label{dyn:backup_reset}
    \EndIf
\EndIf
\vspace{1ex}

\If{$u, v \in \I_{exchange} \AND \Max u < \Max v$} \Comment{exchange maximum} \label{dyn:exchange_cond}
    \State $(\Time u, \Max u, \LastMax u) \gets (\tau_1 \cdot \Max v, \Max v, \LastMax v)$ \label{dyn:exchange}
\EndIf

\vspace{1ex}

\If{$\Max u = \Max v \AND (u \times v) \notin (\I_{exchange} \times \I_{reset})$} \Comment{exchange last maximum} \label{dyn:exchange_last_max_cond}
    \State $\LastMax u \gets \max\{\LastMax u, \LastMax v\}$ \label{dyn:exchange_last_max}
\EndIf

\State $(\Time u, \Interactions u) \gets (\max\{\Time u, \Time v\} - 1, \Interactions u + 1)$ \Comment{CHVP}\label{dyn:chvp}
\EndFunction
\end{algorithmic}
\end{algorithm*}

\section{Dynamic Size Counting Protocol}

\label{chapter:3}
In this section we present the details of our protocol for the loosely-stabilizing size counting problem.
Our protocol is formally specified in \cref{alg:dynamic_size_counting}.
Every agent $u$ stores four variables: \Max u, \LastMax u, \Time u, and \Interactions u.
For newly added agents these variables are initialized 
to $\mathsf{max} = \mathsf{lastMax} = 1$, $\mathsf{time} = \tau_1$, and $\mathsf{interactions} = 0$.
The agents use the variable $\mathsf{max}$ to spread the maximum of GRVs using epidemics.
The $\mathsf{time}$ is synchronized via CHVP.
In contrast, the $\mathsf{interactions}$ variable is not exchanged; it measures the interactions an agent has had since its last reset.

All agents progress through three phases, \emph{exchange}, \emph{hold}, and \emph{reset}, in that order.
For each agent~$u$ we define these phases by dividing $\Time u$ into intervals that depend on the agent's estimate.
In a configuration $C$, agents are in at most one of the following sets corresponding to the respective phases:
\begin{align*}
\I_{exchange} &= \{v \in V: \Time{C(v)} \geq \tau_2  \cdot \Max{C(v)}\}, \\
\I_{hold} &= \{v \in V: \tau_2 \cdot \Max{C(v)} > \Time{C(v)} \geq \tau_3 \cdot \Max{C(v)}\}, \\
\I_{reset} &= \{v \in V: \tau_3 \cdot \Max{C(v)} > \Time{C(v)} \geq 0\}.
\end{align*}
If all agents share the same maximum value $M$ defined as $M = \max_{v \in V}(\max\{\Max {C(v)}, \LastMax{C(v)})\}$, the length of the exchange phase is $|\I_{exchange}| = (\tau_1 - \tau_2) \cdot M$,
the length of the hold phase is $|\I_{hold}| = (\tau_2 - \tau_3) \cdot M$, and
the length of the reset phase is $|\I_{reset}| = \tau_3 \cdot M$.
Here, the $\tau_{i}$ denote large enough constants which we will specify later in \cref{lemma:intervals}.

When agents reset, they generate new GRVs.
The goal of the exchange phase is to spread the maximum of these GRVs to all agents via epidemics.
The hold phase functions as a separator between the exchange phase and the reset phase; its goal is to ensure that an agent that has just left the exchange phase does not immediately reset.
The reset phase is used to launch the agents into the next round while minimizing the spread of the agents' $\mathsf{interactions}$ variables.

Lines \ref{dyn:phase_trans_cond} to \ref{dyn:reset_with_sample} of \cref{alg:dynamic_size_counting} implement the phase transition leading to a reset.
A reset is formally defined as an interaction where the agent sets its $\PlainMax$ to another GRV, its $\PlainTime$ to a multiple of the previous maximum or the new GRV, and its $\PlainInteractions$ to zero.
It follows that resetting agents always enter the exchange phase.
During the exchange phase, agents will adopt larger $\PlainMax$ values without generating new GRVs (see Lines \ref{dyn:exchange_cond},\ref{dyn:exchange}).
In any other cases, resets always lead to new GRVs (see Lines \ref{dyn:reset_with_sample},\ref{dyn:backup_grv}).

If an agent adopts a new maximum, it resets its $\PlainTime$. This is detected by other agents using CHVP (see Line \ref{dyn:chvp}).
If the adoption happened recently and the agent is still in the exchange phase, it stays in the exchange phase.
When agents have not detected a new maximum for a long time, they progress through the hold and reset phases.
Eventually, the agents reset back to the exchange phase.
At this point, resetting agents replace their $\PlainMax$ variable \mbox{with a new GRV.}

Most agents' newly sampled GRVs will be much smaller than $\log n$.
To keep the population synchronized, the agents store a ``trailing'' estimate $\PlainLastMax$.
If the new GRV is larger, they will use it to define the phase lengths.
Otherwise, they will use the last rounds' $\PlainMax$ value (see Line \ref{dyn:reset_with_sample}).
As we will see, this ensures that the phases of all agents stay sufficiently large.

In CHVP agents always adopt larger $\PlainTime$ values of fellow agents (minus 1).
This leads to a problem: 
One agent adopting a $\PlainMax$ might prevent many others storing the same $\PlainMax$ from resetting, and thus also from generating new GRVs.
To prevent this and to guarantee $\Omega(n)$ new GRVs in $O(M + \log n)$ time, we use the $\PlainInteractions$ count to detect such situations.
This effectively guarantees that every agent generates one GRV in $O(M)$ time.
This is described in Lines \ref{dyn:backup_grv_start} to \ref{dyn:backup_reset} of \cref{alg:dynamic_size_counting}.
To preserve synchronization, however, agents will only adopt this ``backup-GRV'' when it is an order of magnitudes larger than their current maximum GRV (see Line \ref{dyn:backup_grv_cond}).
To this end agents ``overestimate'' the saved $\PlainMax$ by $20 (k+1)$ (see Lines \ref{dyn:reset_with_sample},\ref{dyn:backup_reset}).
We use the constant $\tau'$ to control the number of interactions any agent should have before generating a ``backup-GRV''.
We define $\tau'$ in \cref{lemma:intervals}.
Finally, Lines \ref{dyn:exchange_cond} to \ref{dyn:chvp} describe the logic for exchanging $\PlainMax$, $\PlainLastMax$, and $\PlainTime$ values.

\medskip

The intuition for the recovery of our clock from an arbitrary state is the following. As the population starts with an arbitrary maximum $\log \hat{n}$ stored by some agents, the initial maximum might be quite large.
Thus, to generate a new accurate maximum, $\log \hat{n}$ is \emph{forgotten} in $O(\log \hat{n})$ time.
Agents \emph{forget} a maximum when they overwrite it with a different one.
When the entire population forgets a maximum, no agents store it.
Initially, a large $\log \hat{n}$ will dominate the epidemic and immediately replace all other estimates.
Only after the population forgets $\log \hat{n}$, accurate estimates of $\log n$ can spread to the entire population.
Once a new maximum is generated, the agents synchronize and move through the phases together.
They then generate new GRVs, leading to a new valid estimate of $\log n$.
At this point, the population has converged.
This cycle is then repeated again and again, leading to a polynomial holding time w.h.p.
Recall that the phases are defined by time metrics that each agent stores, combined with their current estimate.
Once converged, the phases thus have length $\Theta(\log n)$.

\paragraph{Geometrically Distributed Random Variables}
As in \cite{DBLP:conf/sand/DotyE22} we assume for the sake of our analysis that agents can generate GRVs.
However, this is not a strong assumption.
Indeed, the process of generating one GRV can be split up into multiple interactions, each consisting of one coin flip.
This would allow us, after some warm-up phase, to use synthetic coins as introduced in \cite{DBLP:conf/soda/AlistarhAEGR17}.

To achieve high probability bounds on the maximum of m, Lines~\ref{dyn:phase_trans_cond} to~\ref{dyn:reset_with_sample} of \cref{alg:dynamic_size_counting} implement the previously defined phase transitions leading to a reset.
A reset is defined as the agent setting their $\PlainMax$ to another GRV, their $\PlainTime$ to a multiple of the previous maximum or the new GRV, and their $\PlainInteractions$ to zero.
Thus, resetting agents will always enter the exchange phase.
Agents already in the exchange phase will adopt larger $\PlainMax$ values without generating new GRVs (see Lines \ref{dyn:exchange_cond}, \ref{dyn:exchange}).
In any other cases, resets always lead to new GRVs (see Lines \ref{dyn:reset_with_sample},\ref{dyn:backup_grv}).
As we will see in \cref{lemma:kn_grvs}, $k$ is an arbitrary constant that ultimately allows us to control the error probability of our protocol.
For the sake of simplicity, we will assume that each agent can generate these in one interaction instead of $k$ subsequent interactions.
As $k$ is constant, this does not affect the asymptotic running time complexity.
For completeness, an algorithm that generates the maximum of $k \cdot n$ GRVs is given in \cref{apx:omitted-proofs}.

\section{Analysis}
\label{chapter:5}
In this section we prove \cref{theorem:1,thm:clock}.
Note that we did not optimize the constants for our proofs, and some constants that we use in our analysis are quite substantial (see, e.g., \cref{lemma:intervals}).
However, our empirical analysis in \cref{sec:simulations} shows that the protocol works well with much smaller constants.
We start our analysis in \cref{sec:prelim} with basic observations and notation after which we present in \cref{sec:tools} fundamental building blocks and tools that we use throughout the remainder of this paper.
Then we analyze the convergence time in \cref{sec:conv} and the holding time in \cref{sec:holding}.

\subsection{Preliminaries} \label{sec:prelim}

\paragraph{Global Maximum.}
Agents in the reset phase will reset when interacting with any agent in the exchange phase, no matter what their $\PlainMax$ is.
So in the reset phase, $\PlainMax$ values do not impact the protocol.
In contrast, agents in the wait phase will only reset when interacting with agents in the exchange phase if their $\PlainMax$ differs.
Thus, at any time we define a \emph{global maximum} $M$ as the largest $\PlainMax$ stored by any agent in the exchange or wait phases.
We say a \emph{new global maximum} $M'$ is generated if any agent generates a new GRV such that ${M' > M}$.
A maximum is \emph{forgotten} if it is present only in agents in the reset phase or no agents at all.

Note that the $\PlainMax$ and $\PlainLastMax$ values may differ.
We define all phases using whichever is larger and in the following all formal statements refer to the maximum of $\PlainMax$ and $\PlainLastMax$.
When using $M$ in the context of time complexity, we mean the largest $\PlainMax$ or $\PlainLastMax$ value present in agents with this maximum.
The expression \emph{relatively quickly} means a time complexity of $O(M)$. As we will see, this becomes $O(\log n)$ after some initial convergence time w.h.p.

\paragraph{Synchronized population.}
Our protocols resemble a clock divided into three phases.
The goal of this clock is to synchronize agents around the clock face.
As the clock uses all three variables, our definition of synchronicity also depends on all those variables.
We define a \emph{synchronized population} as a population in a configuration where every agent stores the same value \[\PlainMax, \PlainLastMax \in [0.5 \log n, 40 (k+1)^2 \log n].\]
Furthermore, all agents have to be either in $\I_{exchange} \cup \I_{wait}$ or in $\I_{wait} \cup \I_{reset}$, and all agents must have $\PlainTime < \tau_1 M$.

Lastly, to avoid new ``backup'' GRVs, the $\PlainInteractions$ count must not be too high.
How high exactly depends on how far the population has progressed through the intervals.
Before all agents store the same maximum, the infection process is akin to an epidemic.
Let $T$ denote the time required for an epidemic to finish w.h.p.\ according to \cref{lemma:epidemics}.
Afterward, the agents progress through the three phases.
Let $T'$ denote the time required until the maximum detection value would drop below zero according to \cref{lemma:intervals}.
We can bound how much the $\PlainInteractions$ values may differ depending on these two variables:
At time $t \in [0, T + T']$, no agent may have ${\PlainInteractions > 2t(1+\sqrt{{k}/{t}}) \log n}$.
Thus, when an agent generates a sufficiently large new global maximum that infects all agents, the population synchronizes.
This follows from the definition of \cref{alg:dynamic_size_counting} and is formalized below.

\paragraph{Round.}
The transition from the reset phase back to the exchange phase is explicitly excluded from the synchronicity definition.
When agents reset, they might generate new GRVs smaller than $0.5 \log n$.
However, when the population is synchronized, they will re-synchronize without resetting a second time, in the time required by two epidemics w.h.p.
We call this entire process one \emph{round}.
A round starts at the first synchronized configuration.
It includes the last synchronized configuration and extends until the last interaction before the population synchronizes again.
In this unsynchronized period, each agent must reset exactly once w.h.p.
Thus, once the population is synchronized, agents will directly continue from one round to the next.
This is formalized in \cref{sec:holding} where we analyze the holding time of our dynamic size counting protocol.

\subsection{Toolbox} \label{sec:tools}

In this section, we provide a collection of properties and tools, most of which are from related works.
These properties are then used to define and analyze our protocol rigorously.

\paragraph{Maximum of Geometric Random Variables}

Whenever we refer to GRVs we mean random variables with geometric distribution $\Geom(1/2)$.
It is folklore that the maximum of $n$ independent and identically distributed GRVs is concentrated around $\Theta(\log n)$, see, e.g., Lemma D.7 in \cite{DBLP:conf/podc/DotyE19}.
The following result is a straight-forward extension.
We give the proof in \cref{apx:omitted-proofs}.

\begin{lemma}[restate=restatekngrvs,label={lemma:kn_grvs}]
Let $k \geq 1$ be an arbitrary constant with $k \leq n$ and $n \geq 50$, and let $G = \{G_1, G_2, \dots, G_{k\cdot n}\}$ be a set of $k\cdot n$ i.i.d.\ random variables with distribution $\Geom(1/2)$. Define  $M = \max\{G\}$.
Then \[
Pr[0.5 \log n \leq M \leq 2(k+1) \log n] \geq 1 - O(n^{-k}).
\]
\end{lemma}

\paragraph{Epidemics}
In epidemics, agents store a single value and adopt the maximum of any agent's value they encounter.
Formally, the transition rule reads
$(u, v) \rightarrow (\max\{u,v\}, v)$.
Assume that initially one agent is in state $1$ and $n-1$ agents are in State $0$.
It is folklore that within $O(n \log n)$ interactions every agent has state $1$ w.h.p., see, e.g., Lemma 2 in \cite{DBLP:journals/dc/AngluinAE08a}.
The following variant allows us to control the error probability by providing a high-probability bound in terms of $k$ (see, e.g., Lemma 4 in \cite{DBLP:conf/sand/BerenbrinkBHK22}).

\begin{lemma}
    \label{lemma:epidemics}
    Let $k$ be any positive constant and let $M$ be defined as $M = \max_{v \in V}\{C_0(v)\}$.
    After at most $t \leq 4(k+1) n \log n$ interactions every agent is in state $M$ with probability at least $1 - O(n^{-k})$.
\end{lemma}

\paragraph{Detection Protocol}
\Textcite{DBLP:conf/wdag/SudoEIM21} use the detection protocol from \cite{DBLP:conf/dna/AlistarhDKSU17} by counting down from a starting value.
The agents exchange the maximum of their values and decrement it by one. 
This happens until their values reach zero.
In our protocol we use a one-sided version of this CHVP process.
The transition rule is defined as \[{(u, v) \rightarrow (\max\{u, v\} - 1, v)}. \]

The following lemma is based on Lemma~3 from \cite{DBLP:conf/wdag/SudoEIM21}, which is in turn based on Lemma~1 from~\cite{DBLP:conf/dna/AlistarhDKSU17}.
We give a detailed overview of one-sided CHVP in \cref{apx:chvp}.

\begin{lemma}[restate=restatechvpmax,label={lemma:chvp_max}]
    Let $m = \max_{v \in V}\{C_0(v)\}$, $\Delta$ be an arbitrary positive integer, and $k$ be a positive constant.
    Any execution of CHVP enters a configuration $C_\tau$ with 
    ${\tau < 7n(\Delta + k \log n)}$%
    , where ${\max_{v \in V}\{C_\tau(v)\} \leq m-\Delta}$ holds with probability at least $1 - n^{-k}$.
\end{lemma}

In the following lemma we investigate the lower bound of the CHVP variable, which is based on Lemma~4 from \cite{DBLP:conf/wdag/SudoEIM21}.
Initially, the minimum value can be arbitrarily small.
We now show that after $O(\log n)$ time, this minimum will be close to the initial maximum $m$.
To do this, we model the CHVP process as an epidemic starting from $m$.
All agents starting with $m$ start off as infected, with the rest being uninfected.
Once all agents have been infected, their minimum value will be bounded by the number of interactions each agent initiates.

\begin{lemma}[restate=restatechvpmin,label={lemma:chvp_min}]
    Let $m = \max_{v \in V} \{C_0(v)\}$, $\Delta$ be an arbitrary positive integer, and $k \geq 2$ be a constant.
    Any execution of CHVP enters a configuration $C_\tau$ with $\tau = 7n(\Delta + k \log n)$, where $\min_{v \in V} \{C_\tau(v)\} \geq m - 12 (\Delta + k \log n)$ holds with probability at least $1 - n^{-k}$.    
\end{lemma}

For a configuration $C$ we define $C[V] = \{ C(v) \mid v \in V\}$. Then a configuration $C$ of CHVP lies in some interval $I$ if $C[V] \subset I$. With slight abuse of notation, we will omit the $[V]$ and write $C \subset I$ meaning that all values of all agents in $C$ lie in that interval.
In the following lemma, we provide the actual constants for which our analysis holds. Recall that we did not try to optimize these constants but they have been chosen for mere convenience.

\begin{lemma}
\label{lemma:intervals}
    Let $k \geq 2$ be a constant, $\tau_1 = 1140 k$, $\tau_2 = 1119 k$, $\tau_3 = 454 k$, and $\tau' = 4350 k$.
    Then for any $M \geq 0.5 \log n$ and $\max_{v \in V}\{C_0(v)\} = \tau_1 M$, there exist $i_1 < i_2 < i_3 = O(nM)$, such that $C_{i_1} \in (\tau_1 M, \tau_2 M], C_{i_2} \in (\tau_2 M, \tau_3 M], C_{i_3} \in (\tau_3 M, 0]$ before any agent initiated more that $\tau' M$ interactions with probability at least $1 - O(n^{-k})$.
\end{lemma}
\begin{proof}
    We define $m$ such that $M = m \log n$ and $i_1 = 8 n (k+1) m \log n$.
    According to \cref{lemma:chvp_min},
    for $\Delta_1 = (8/7 (k+1) m - k) \log n$, it holds that {
    \medmuskip=2.0mu plus 2.0mu minus 2.0mu
\thinmuskip=2.0mu
\thickmuskip=2.0mu plus 5.0mu
    \begin{align*}
        C_{i_1} &\subset (\tau_1 m \log n, \tau_1 m \log n - 12 (\Delta_1 + k \log n)] \\
        &\subset (\tau_1 m \log n, (\tau_1 - 21 k) m \log n] = (\tau_1 M, \tau_2 M].
    \end{align*}
    }
    Note that $i_1$ allows for two epidemics to complete w.h.p.\ (see \cref{lemma:epidemics}).
    All agents have left the first interval after $i_2 = 400n km \log n \geq i_1 + 7n(21 km + k) \log n $ interactions (see \cref{lemma:chvp_max}).
    Thus, for $\Delta_2 = (400/7 m - 1) k \log n$, it holds that
    \begin{align*}
    C_{i_2} &\subset ((\tau_1 - 21k) m \log n, \tau_1 m \log n - 12 (\Delta_2 + k \log n)] \\
    &\subset ((\tau_1 - 21 k) m \log n, (\tau_1 - 686 k) m \log n] = (\tau_2 M, \tau_3 M].
    \end{align*}
    Lastly, all agents will have left this interval after $i_3 = 1065n km \log n = 665 n km \log n + i_2$ interactions.
    For the lower bound, we now have a concrete desired value: $\min_{v \in V}\{C_{i_3}\} \geq 0$.
    During $i_3 - i_2$ interactions the minimum reaches at most $\tau_1 m \log n - 1140 km \log n$.
    Thus, $\tau_1 = 1140 k$, $\tau_2 = 1119 k$, and $\tau_3 = 454 k$.
    Then 
    \begin{align*}
    C_{i_3} &\subset (\tau_1 m \log n - 1140 k m \log n, 0]  = (\tau_3 M, 0].
    \end{align*}
    As the probability of failure for both \cref{lemma:chvp_min,lemma:chvp_max} is $n^{-k}$, by applying a union bound this holds with probability at least $1 - O(n^{-k})$.
    Lastly, all agents will reach time zero or reset beforehand after at most $i_3 + 7n(454 km + k) \log n \leq 4257 km \log n$ interactions w.h.p.\ (see \cref{lemma:chvp_max}), initiating at most
    \begin{align*}
    4257 k m \left(1+\sqrt{1/(4257m)}\right) \log n 
    &\leq 4350 k m \log n = \tau' M
    \end{align*}
    interactions with probability $1 - n^{-k}$ (see \cref{lemma:agent_participations} in \cref{apx:omitted-proofs}).
\end{proof}

If $M \geq 0.5 \log n$ and $\max_{v \in V}\{C_0(v)\} = \tau_1 M$,
this lemma implies that there exist $\tau_1, \tau_2, \tau_3$ 
such that in $\Theta(M)$ time all agents traverse through intervals 
$(\tau_1 M, \tau_2 M], (\tau_2 M, \tau_3 M], (\tau_3 M, 0]$ 
together w.h.p.
Additionally, $\tau_1 - \tau_2$ is long enough for two epidemics to complete before any agent exits the first interval.
Thus, the $\PlainTime$ propagated via CHVP fulfills the requirements of our clock.

\medskip

This concludes the presentation of the different tools that we use in our protocols.
When applied accordingly, they provide the means to approximate the population size, exchange those estimates, and define synchronized clocks.
Using these building blocks, we will now define our loosely-stabilizing, uniform phase clock to solve the dynamic size counting problem.

\subsection{Convergence Time} \label{sec:conv}
One fundamental observation is that new global maxima lead to a synchronized population.
When applying this to synchronized populations, all agents enter the exchange phase together, thereby effectively forgetting any old maxima.
Thus, the first resetting agent generates a new global maximum.
Similarly, in unsynchronized populations, agents will forget the initial maxima and generate new ones relatively quickly, thus synchronizing the population.
This is formalized in the following lemma.

\begin{lemma}
    \label{lemma:new_m_syncs}
    Let $k \geq 2$ be a constant. When any agent generates a new global maximum $M \geq 0.5 \log n$, it synchronizes the entire population in $O(M + \log n)$ time with probability $1 - O(n^{-k})$.
\end{lemma}
\begin{proof}
    As $M$ is new, the generating agent is in the exchange phase with $\PlainTime \geq \tau_1 M$.
    The new maximum will spread via epidemic.
    All infected agents will stay in the exchange phase until the entire population stores the same maximum, as required by the definition of $\tau_1$.
    Once infected, the agents will stay in the exchange phase, with $\PlainMax=M$ unless a new, even larger maximum is generated.
    In the worst case, the final infected agent generates a new, larger GRV $M'$.
    As this would represent the overestimated maximum of $kn$ GRVs, $M' \in [10 (k+1) \log n, 40 (k+1)^2 \log n]$ holds w.h.p.

    However, the $\PlainInteractions$ values might not be synchronized yet, leading to new ``backup'' GRVs.
    Initially, most agents might store $\PlainMax \in O(1)$, leading to $O(n \log n)$ interactions while $M$ still spreads.
    At most $4 (k+1) k n \log n$ GRVs are generated during this time (see \cref{lemma:epidemics}).
    Using \cref{lemma:kn_grvs}, we can provide the upper bound of their maximum as
    \begin{align*}
        \MoveEqLeft 2 (k+1) \log (4 (k+1) k n \log n) \\
        &= 2 (k+1) (2 + \log (k+1) + \log (k) + \log n + \log \log n) \\
        &\leq 2 (k+1) (2 + 4 \log n) < 10 (k+1) \log n
    \end{align*}
    with probability at least $1 - O(n^{-k})$.
    After that point, each agent will store $\PlainMax \in \Omega(M)$.
    Thus, by the definition of $\tau'$, each agent will generate at most $k$ GRVs before the clock is traversed, unless a new maximum is generated.
    The maximum of these GRVs is at most $10 (k+1) \log n$ with probability at least $1 - O(n^{-k})$.
    If the new maximum is indeed larger, it would spread to all agents.
    Since it would have to be larger than $0.5 \log n$ and thus be overestimated to at least $0.5 \cdot 20 (k + 1) \log n = 10 (k+1) \log n$, no subsequent ``backup'' GRVs would be large enough to replace the new maximum and disrupt the synchronization again.
    Either way, all agents will store the same $\PlainMax$ value without a larger one being generated during the clock's traversal w.h.p.

    As all agents now progress through all phases together (see \cref{lemma:intervals}), and have a $\PlainMax \geq 0.5 \log n$, the ${\text{reset} \rightarrow \text{exchange}}$ phase transition spreads via epidemic, while no agent leaves the exchange phase.
    Thus, exactly $kn$ GRVs are generated, with their overestimated maximum being in $[10 (k+1) \log n, 40 (k+1)^2 \log n]$, and spreading to the entire population.
    At this point, the $\PlainInteractions$ values will differ by at most $4 (k+1) \log n (1 + \sqrt{k/(4 (k+1) \log n)})$.
    Note that $\PlainLastMax$ can still be $\omega(\log n)$ here.
    However, as no new ``backup'' GRVs will be generated during the traversal of the clock, all agents will traverse through it together.
    After another such round taking $O(M)$ time, it holds that $\PlainMax, \PlainLastMax \in \Theta(\log n)$.
    Thus, the clock is synchronous.
\end{proof}

In the following, we will show that agents will always generate a new global maximum in $O(M + \log n)$ time.
Together with \cref{lemma:new_m_syncs}, this proves that agents will synchronize within this time.
Starting with an $M$ that is low, the agents will generate a new global maximum in polylogarithmic time.
\begin{lemma}
    \label{lemma:small_m_new_max}
    Let $C_0$ be an arbitrary initialization with $M < 0.5 \log n$.
    Then a new global maximum $M' \geq 0.5 \log n$ will be generated in $O(\log n)$ time w.h.p.
\end{lemma}
\begin{proof}
    Observe that in $c \log n$ time, agents have $\Theta(c \log n)$ interactions w.h.p.\ (see \cref{apx:omitted-proofs}).
    Thus, in $O(M) \leq O(\log n)$ time, the ``backup'' GRV generation is triggered for all agents, resulting in at least $k n$ GRVs, with the maximum being at least $0.5 \log n$.
\end{proof}

In the following lemma, we show that starting from large $M$, the agents generate a new global maximum relatively quickly.

\begin{lemma} 
    \label{lemma:large_m_new_max}
    Let $C_0$ be an arbitrary initialization, with $M \geq 0.5 \log n$.
    Then a new global maximum will be generated in $O(M + \log n)$ time w.h.p.
\end{lemma}

The idea is as follows.
Agents storing $M$ will infect other agents, thereby triggering an epidemic and almost synchronizing the population, except for their $\PlainInteractions$ values.
The agents will then end up in the reset phase together and eventually reset, leading to the population forgetting $M$.
If a new, larger maximum is generated at any time during this process, this will synchronize the population directly.

\begin{proof}
    For the sake of simplicity, we will assume that at least one agent storing $M$ infects another agent.
    While this is no prerequisite, $M$ provides an upper bound in the time complexities for all maxima at least $0.5 \log n$.
    
    We require this specific distinction because of one edge case.
    Assume a subset of agents are initialized with $M$, such that almost all are in the reset phase.
    However, the rest are $O(1)$ interactions away from entering it as well.
    Depending on how close those agents are to the reset phase and how many agents with $M$ are in the reset phase already, the probability of them infecting any agents with $\PlainMax < M$ diminishes.
    When no new agent is infected with such a $\PlainMax$, all agents storing $\PlainMax \geq 0.5 \log n$ will traverse through all phases in $O(M)$ time.
    Afterward, by \cref{lemma:small_m_new_max}, they will synchronize in $O(\log n)$ time.
    
    As we assume one agent will be infected with $M$ and thus have $\PlainTime=\tau_1 M$, this will trigger an epidemic spreading of $M$.
    All agents will be infected in ${O(\log n) \leq O(M)}$ time (see \cref{lemma:epidemics}) and then progress through the phases almost synchronously.
    We say almost because the $\PlainInteractions$ values are not yet synchronized.
    \Cref{lemma:new_m_syncs} holds if at any time agents generate GRVs larger than $M$.
    Then, the population synchronizes relatively quickly.
    Otherwise, the population will end up together in the reset phase, as required by the definition of $\tau_3$ (see \cref{lemma:intervals}).
    Once the first agent resets, they will trigger the ${\text{reset} \rightarrow \text{exchange}}$ exchange phase transition epidemic.
    As all resetting agents will set their $\PlainTime$ to at least $\tau_1 M$, they will not exit this phase before all agents enter it and spread the new maximum.
    As $n$ resetting agents generate $k$ samples each, the new overestimated maximum of those GRVs is in ${[10 (k+1) \log n, 40(k+1)^2 \log n]}$ w.h.p.\ (see \cref{lemma:kn_grvs}).
    Thus, \cref{lemma:new_m_syncs} holds, proving the clock synchronizes.
\end{proof}

We have shown that the agents will generate a new global maximum starting from any initial maxima.
As new global maxima synchronize the population, any arbitrary initialization will become synchronized.

\begin{lemma}
    \label{lemma:convergence_time}
    The population synchronizes in $O(M + \log n)$ time from an arbitrary initialization.
\end{lemma}

\begin{proof}
    Agents generate a new global maximum in $O(M + \log n)$ time.
    This follows directly from the combination of the time bounds from \cref{lemma:small_m_new_max,lemma:large_m_new_max}.
    Either one of the two cases has to hold for any given $M$.
    By \cref{lemma:new_m_syncs}, such a new global maximum synchronizes the population in $O(M + \log n)$ time.
\end{proof}

\subsection{Holding Time} \label{sec:holding}
For the holding time, we assume the first synchronized configuration of a round as the initialization.
First, we show that starting from this configuration, one round directly continues to the next w.h.p.
Then we apply the union bound to determine the holding time w.h.p.

\begin{lemma}
    \label{lemma:stays_synced}
    Once synchronized, the population stays synchronized until the first agent resets w.h.p.
\end{lemma}

\begin{proof}
    As $\tau_1, \tau_2, \tau_3$ are defined according to \cref{lemma:intervals}, the agents will gather in each interval w.h.p.
    Meanwhile, no agent will have more than $\tau' M$ interactions w.h.p.\ due to $\tau'$ being defined in accordance with \cref{lemma:intervals}.
    Thereby, no new ``backup'' GRVs will be generated.
    All agents will keep $M$ until the first agent reaches $\PlainTime = 0$ and thus resets.
\end{proof}

The following lemma shows that when a population is synchronized, agents will eventually reset exactly once before becoming synchronized again.

\begin{lemma} 
    \label{lemma:round_to_round}
    Let $C_0$ be a synchronized configuration.
    This synchronization implies a round and will continue directly into another round w.h.p.
\end{lemma}

In the proof, we show that eventually no agents will be left in the exchange or hold phases.
This means the first resetting agent generates a new global maximum.
Then we can reduce this problem to \cref{lemma:new_m_syncs}, without any ``backup'' GRVs being generated.

\begin{proof}
    It follows directly from \cref{lemma:stays_synced} that all agents stay synchronized until the first one resets.
    Before the first agent resets, all agents are in the reset phase, as required by $\tau_{3}$ (see \cref{lemma:intervals}).
    Thus, the first resetting agent will generate a new global maximum among the exchange and hold phases.
    While this might be smaller than $0.5 \log n$, the agent will still stay in the exchange phase for $\Theta(\log n)$ time due to the trailing estimate.
    During this time, they will trigger the ${\text{reset} \rightarrow \text{exchange}}$ epidemic while staying in the exchange phase themselves.

    After every agent has entered the exchange phase, $kn$ GRVs have been generated.
    The maximum of those is at least $0.5 \log n$ w.h.p.
    With this maximum, \cref{lemma:new_m_syncs} holds.
    As all agents stay in the exchange phase until this second epidemic has finished, no further samples are generated, but all agents will become synchronized again.
    Thus, the population synchronizes again w.h.p.
\end{proof}

The following lemma puts everything together and shows the holding time for the Adoption Detection protocol.

\begin{lemma}
    \label{lemma:holding_time}
    Starting from a synchronized configuration, the population will continue directly from round to round for $\Theta(n^{k-1} \log n)$ time w.h.p.
\end{lemma}

\begin{proof}
    From \cref{lemma:round_to_round}, it follows that starting from a synchronized configuration, the population continues from one round to the next with probability at least $1 - O(n^{-k})$.
    By the union bound over $n^{k-1}$ rounds, the probability that all continue successively is at least $1 - O(n^{-1})$.
    As all agents reset exactly once, any new global maximum is $\Theta(\log n)$ w.h.p.\ by \cref{lemma:kn_grvs}.
    Together with the requirements for the phase intervals (see \cref{lemma:intervals}), it follows that one round requires $\Theta(\log n)$ time w.h.p.
    Thus, $n^{k-1}$ rounds require $\Theta(n^{k-1}\log n)$ time to complete w.h.p.
\end{proof}

\subsection{Proof of the Theorems}
In this section, we put everything together and prove our main result, \cref{theorem:1}.
We start with the analysis of the space complexity.
\begin{lemma}
    \label{lemma:space}
    Let $C_0$ be an arbitrary initialization, with $s$ denoting the largest value stored in any of the agents' variables.
    Then \cref{alg:dynamic_size_counting} requires $O(\log s + \log \log n)$ bits per agent w.h.p.
\end{lemma}
\begin{proof}
    All variables of agents store asymptotically equivalent values.
    Let $M$ denote the maximum value generated by our protocol during any execution.
    The $\PlainMax$ and $\PlainLastMax$ variables are thus bound by $M$.
    The $\PlainTime$ variable is bounded by $\tau_1 M = O(M)$.
    Similarly, the $\PlainInteractions$ variable resets back to $0$ after $\tau' M = O(M)$.
    
    As we assume a constant $k$, during each round, the largest maximum will be $O(\log n)$ w.h.p.\ according to \cref{lemma:kn_grvs}.
    By the union bound, in polynomial time, the largest maximum will still be $O(\log n)$ w.h.p.
    As all these values can be stored in binary representation, each variable requires $O(\log s + \log \log n)$ bits.
    Thus, the space requirements are $4 O(\log s + \log \log n) = O(\log s + \log \log n)$ bits w.h.p.
\end{proof}

Finally, we will now give the proofs of \cref{theorem:1,thm:clock}.
\begin{proof}[Proof of \cref{theorem:1}]
    The proof of \cref{theorem:1} follows from \cref{lemma:holding_time} for the holding time, \cref{lemma:convergence_time} for the convergence time, and \cref{lemma:space} for the space complexity.
\end{proof}

\begin{proof}[Proof of \cref{thm:clock}]
By assumption of the theorem, all agents are in a synchronized configuration at time $t_0$.
From \cref{lemma:round_to_round} it follows that the agents will perform one synchronized round.
Observe that the reset is triggered by an epidemic that concludes in $O(n \log n)$ interactions while the constants that control the length of a round are set to a multiple of that time.
The theorem therefore follows immediately from \cref{lemma:holding_time}.
\end{proof}

\section{Empirical Analysis}
\label{sec:simulations}

\begin{figure*}[t]
\centering

\begin{minipage}{0.5\textwidth}
\includegraphics[width=\textwidth]{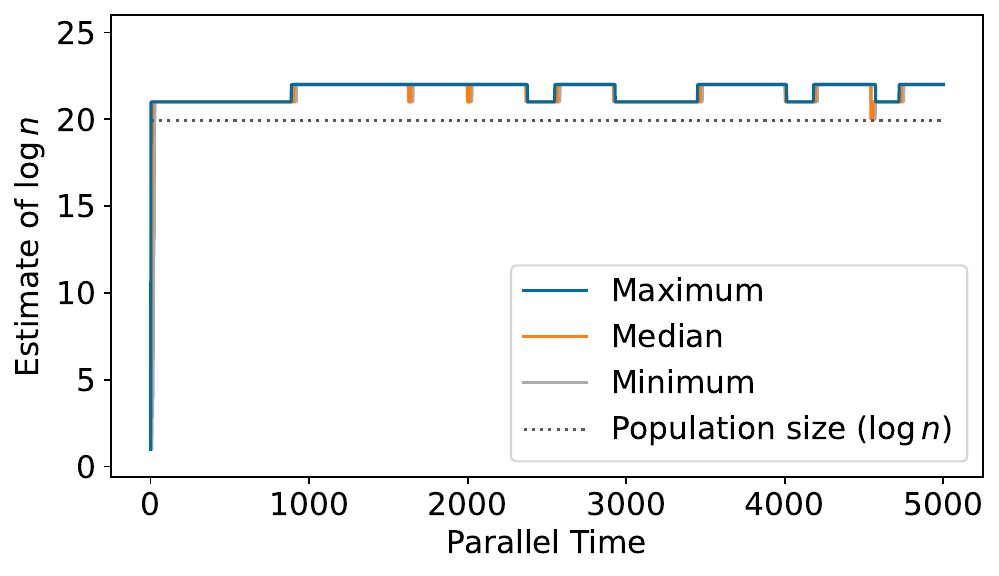}\label{fig:adopt_example_estimate_rand_1}
\caption{Size estimate in a system of $10^6$ agents}
\Description{A plot showing the size estimate in a system of $10^6$ agents}
\end{minipage}\begin{minipage}{0.5\textwidth}
\includegraphics[width=\textwidth]{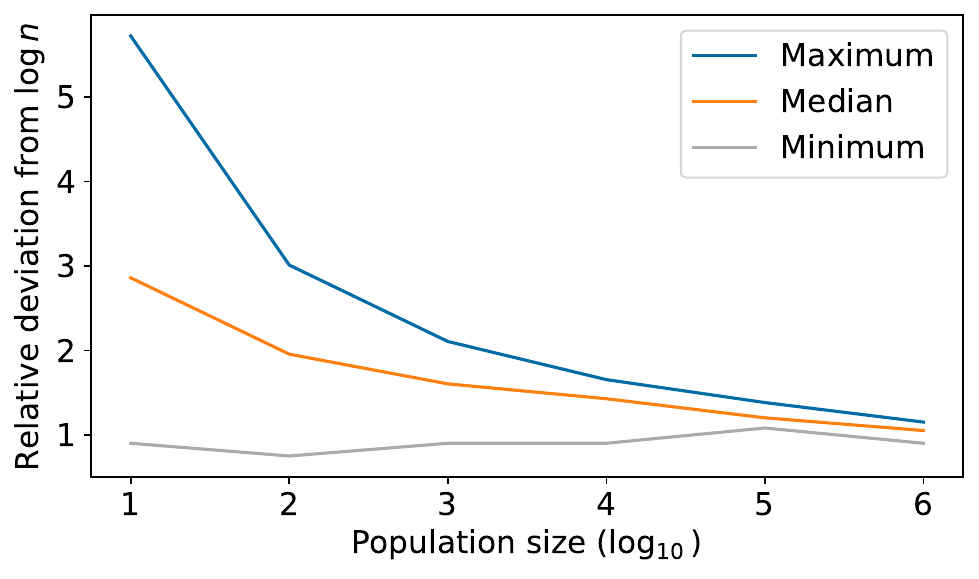}
\caption{Relative error for different values of $n$}
\Description{A plot showing the relative error for different values of $n$}
\label{fig:adopt_relative_deviation}
\end{minipage}

\bigskip

\includegraphics[width=\textwidth,trim=0pt 0pt 0pt 7.0cm, clip]{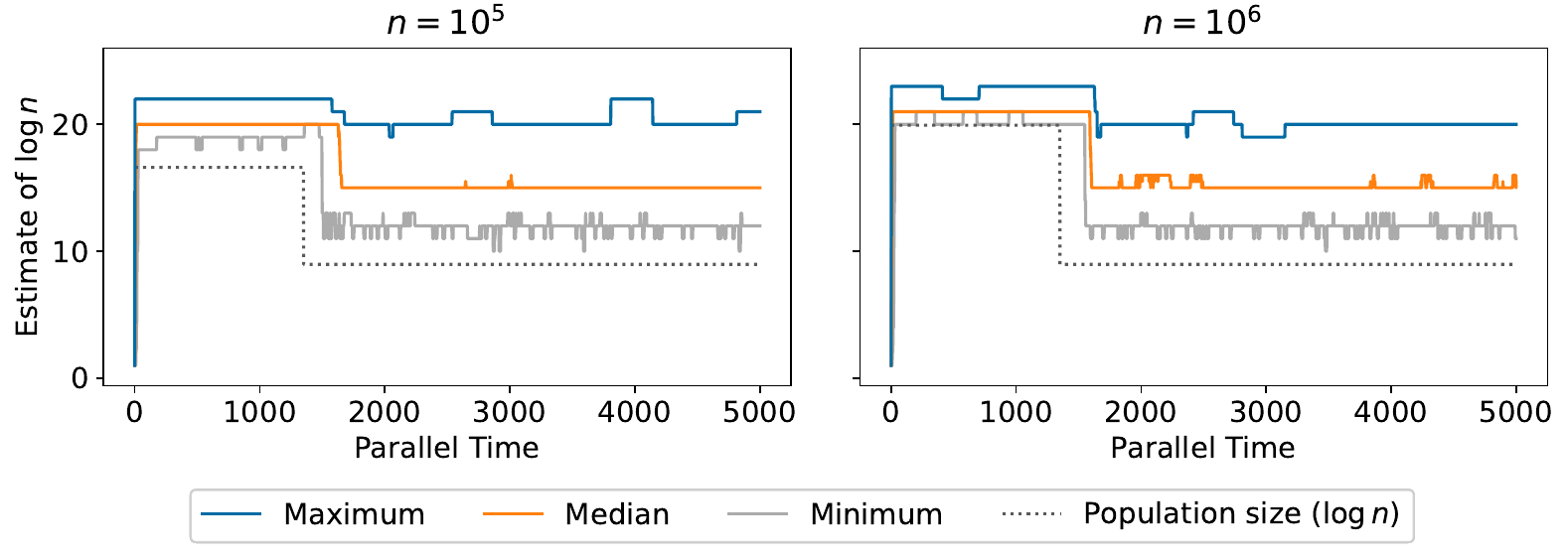}
\caption{Size estimate for different population sizes. 
All but $500$ agents are removed after $1350$ parallel time. }
\Description[A plot showing size estimate for different population sizes]{Size estimate for different population sizes. 
All but $500$ agents are removed after $1350$ parallel time. In our simulations the reported estimate of an agent $u$ is $\max\{\Max u, \LastMax u\}$ without the over-estimation applied in \cref{alg:dynamic_size_counting}. }
\label{fig:adopt_dynamic}
\end{figure*}

In this section, we present empirical data and show that our protocol works well for practical instance sizes.
In particular, our data confirm that our protocol requires only modest constant factors.
Our protocol is parameterized by $\tau_1, \tau_2, \tau_3$ to define the phase lengths.
In our simulation we set the constants to $\tau_{1}=6$, $\tau_{2}=4$, and $\tau_{3}=2$.
Additionally we set $\tau'=20$ and $k=16$, resulting in a theoretical holding time of $\Omega(n^{15})$.
In our simulations the reported estimate of an agent $u$ is $\max\{\Max u, \LastMax u\}$ without the overestimation applied in \cref{alg:dynamic_size_counting}.
Unfortunately, our protocol has an unbounded state space hence we cannot use ready-made simulators \cite{DBLP:conf/esa/BerenbrinkHK0PT20,DBLP:conf/cmsb/DotyS21}.
We therefore implement a custom simulator software using the C++ programming language.
As a source of randomness we use the pseudo-random number generator (PRNG) ranlux \cite{LUSCHER1994100}, a high quality PRNG with a period length of over $10^{100}$ \cite{james_review_2020}.
To ensure independence among our simulations we seed the PRNG with a non-deterministic random number using the C++ random\_device generator.
Our simulator software is available from our public GitHub repository.\footnote{\url{https://github.com/dcmx/DynamicSizeCounting}}

We simulate populations of up to $10^6$ agents for $5000$ parallel time steps.
Each data point is generated from 96 independent simulation runs.
To ensure quick simulation times, we do not report the configuration after each interaction but instead we create a snapshot every $n$ interactions.
In our first plot in \cref{fig:adopt_example_estimate_rand_1} we assume that the system is initially empty and show the minimum, median, and maximum values of all 96 estimates over 5000 parallel rounds for $n = 10^6$.
Then in \cref{fig:adopt_relative_deviation} we plot the relative error for varying population sizes $n = 10^1, 10^2, \dots, 10^6$. Together, the first two plots confirm our theoretical results by showing that our algorithm has a long holding time and the resulting estimate approaches $\log n$ as $n$ grows.
Similarly to \cite{DBLP:conf/sand/DotyE22} we also present in \cref{fig:adopt_dynamic} simulation data where an adversary changes the population size during the simulation.
To this end, we reduce the population size to $500$ after $1350$ parallel time.
Our data show the rapid adaption of the estimate to the new population size.
Note that the estimates of the decimated populations deviate a lot, in particular the maximum values.
This conforms to the findings regarding relative deviation from \cref{fig:adopt_relative_deviation}.
Nevertheless, the drop is clearly visible in the estimates, particularly for larger values of $n$.

\section{Conclusion}
\label{chapter:6}

We present a new protocol that solves the dynamic size counting problem.
In addition, our protocol constitutes a loosely-stabilizing uniform phase clock.
The main reason why our protocol can be used to synchronize the agents into phases is its inherent oscillating behavior.
Intuitively, our protocol uses similar phase transitions as the loosely-stabilizing phase clock by \textcite{DBLP:conf/sand/BerenbrinkBHK22}.
This is the main difference to the work by \textcite{DBLP:conf/sand/DotyE22}; their protocol uses the detection protocol by \textcite{DBLP:conf/dna/AlistarhDKSU17} in a continuous fashion.
We believe that our clocks are of independent interest:
loosely-stabilizing phase clocks have applications as an underlying synchronization mechanism for further loosely-stabilizing and dynamic population protocols.
We highlight that it is an intriguing open question to design a loosely-stabilizing phase clock that uses only $o(\log n)$ (or even only constantly many) states.

Our protocol has a significantly reduced space complexity compared to the best previously known protocol for this problem \cite{DBLP:conf/sand/DotyE22}.
This improvement comes at the expense of possibly a larger convergence time.
We remark, however, that from an asymptotic point of view this increase in the convergence time only ``strikes'' if the system is initialized with an estimate $\hat n$ that exponentially over-estimates the true population size $n$.
If $\hat n$ is within polynomial bounds of $n$, our convergence time remains asymptotically the same.
It is an open problem to formally prove a non-trivial trade-off between the required number of states and the convergence time.

Regarding the quality of the approximation, \textcite{DBLP:conf/podc/DotyE19} use in the static setting the average of $O(\log n)$ maxima of $n$ GRVs each.
This leads to an additive factor approximation of $\log n$.
It is an open question whether a similar extension of our protocol could also provide agents with a more accurate estimate.
\Textcite{DBLP:journals/tcs/DotyE21} furthermore show that non-uniform protocols can be made uniform by composing them with a size counting protocol.
They assume a fixed population size and thus use the size counting protocol only once.
In dynamic populations, non-uniform protocols must be restarted every time the size changes.
A formal analysis of a general framework that allows to compose dynamic size counting protocols with non-uniform protocols in the dynamic setting is an open problem.

\begin{acks}
The financial support by the Austrian Federal Ministry for Digital and Economic Affairs, the National Foundation for Research, Technology and Development, as well as the Christian Doppler Research Association is gratefully acknowledged.
\end{acks}

\printbibliography

\clearpage
\onecolumn

%\end{document}
\appendix

\section*{Appendix}

\section{Omitted Proofs and Additional Details}
\label{apx:omitted-proofs}

\subsubsection*{Geometric Random Variables}

The following algorithm can be used by an agent to generate the maximum of $k$ GRVs.

\begin{algorithm}[H]
\caption{GRV($k$), executed when an agent generates the maximum of $k$ new GRVs.}
\label{alg:grv}
\begin{algorithmic}[1]
\Function{GRV}{$k$}
\State $M \gets 1$
\For{$i = 1 \text{ to } k$} \label{line:k_grvs}
    \State $grv \gets 1$ \label{line:grv_init}
    \While{a fair coin lands on heads} \label{line:grv_loop}
        \State $grv \gets grv + 1$ \label{line:grv_increase}
    \EndWhile
    \State $M \gets \max\{grv, M\}$ \label{line:max_k_grvs}
\EndFor
\State \Return{$M$}
\EndFunction
\end{algorithmic}
\end{algorithm}

\subsubsection*{Expected Number of Interactions.}
One unit of parallel time consists of $n$ interactions.
Within $c\log n$ time, every agent is expected to initiate $\Theta(c\log n)$ interactions.
We obtain the following lemma by substituting $\delta = \sqrt{{k}/{c}}$ in Prop.\ 1 from \cite{DBLP:journals/eatcs/ElsasserR18} and adapting it to one-way interactions.

\begin{lemma}[Prop. 1 in \cite{DBLP:journals/eatcs/ElsasserR18}]
    \label{lemma:agent_participations}
    Let $k, c$ be positive integers with $k < c$. 
    Within $c \log n$ parallel time, each agent initiates between ${c(1-\sqrt{{k}/{c}})\log n}$ and ${c(1+\sqrt{{k}/{c}})\log n}$ interactions with probability at least $1 - n^{-\Theta(k)}$.
\end{lemma}

\subsubsection*{Maximum of Geometric Random Variables}

The following lemma is used in the proof of \cref{lemma:kn_grvs}.

\begin{lemma}[Lemma D.7 in \cite{DBLP:conf/podc/DotyE19}]
    \label{lemma:n_grvs}
    Let $G = \{G_1, G_2, \dots, G_n\}$ be a set of $n \geq 50$ i.i.d.\ GRVs and define $M = \max_{0 < I \leq |G|}\{G_i\}$.
    Then \[
    \Pr[0.5 \log n \leq M \leq 2 \log n] \geq 1 - n^{-k}
    \]
\end{lemma}

We now prove \cref{lemma:kn_grvs}.

\restatekngrvs*

\begin{proof}

    For the lower bound, we show that
    $\Pr[M \geq 0.5 \log n] \geq 1 - n^{-k}.$
        We assume $k n$ i.i.d.\ GRVs.
        This is equivalent to looking at the maximum of $k$ independent maxima of $n$ GRVs.
        As these trials are independent, we can use the probability for each trial from \cref{lemma:n_grvs}.
        Let $E_i$ be the event that $\max_{j \in \left[i \cdot n, (i+1) \cdot n\right)} \{G_i\}
        < 0.5 \log n$.
        Then we get 
        \begin{align*}
        Pr\left[M \geq 0.5 \log n\right] &= 1 - \prod_{i = 0}^{k} Pr\left[E_i\right] \\
        &= 1 - \left(n^{-1}\right)^k \\
        &= 1 - n^{-k}.
        \end{align*}
        
    For the upper bound,
        we will assume another set containing $k n^k$ i.i.d.\ GRVs and reference their maximum as $M'$.
        As $k \geq 1$, it holds that $k n^k \geq k n$.
        Thus, the upper bound for $M'$ also bounds $M$ and its probability.
        Together with \cref{lemma:n_grvs}, we get 
        \begin{align*}
        Pr\left[M' \leq 2 \log (k \cdot n^k)\right] &= Pr\left[M' \leq 3 k \ln n + 2 \log k \right] \\
        &= 1-(k \cdot n^k)^{-1} \\
        &= 1-k^{-1} n^{-k} \\
        &\geq 1 - n^{-k}.
        \end{align*}
        As $2 k \log n + 2 \log k \leq 2 (k+1) \log n$ for all $n \geq k$, this simplification provides a sufficient upper bound.
        Thus, $$Pr\left[M \leq 2 (k+1) \log n \right] \geq Pr\left[M' \leq 2 (k+1) \log n \right] \geq 1 - n^{-k}.$$

    By the union bound, the probability of failure is at most $2 n^{-k} = O(n^{-k})$. Hence 
    \begin{align*}
        Pr\left[(0.5 \log n) \leq M \leq (2(k+1) \log n)\right] \geq 1 - O(n^{-k}).
        \tag*{\qedhere}
    \end{align*}
\end{proof}

\clearpage
\section{Additional Simulation Results}

\Cref{fig:adopt_60_results} shows the minimum, median, and maximum estimate values of populations initialized with an estimate of $60$. The data show that for small population sizes the initial estimate indeed dominates the convergence time.

\begin{figure}[H]
  \centering
  \includegraphics[width=1.02\linewidth]{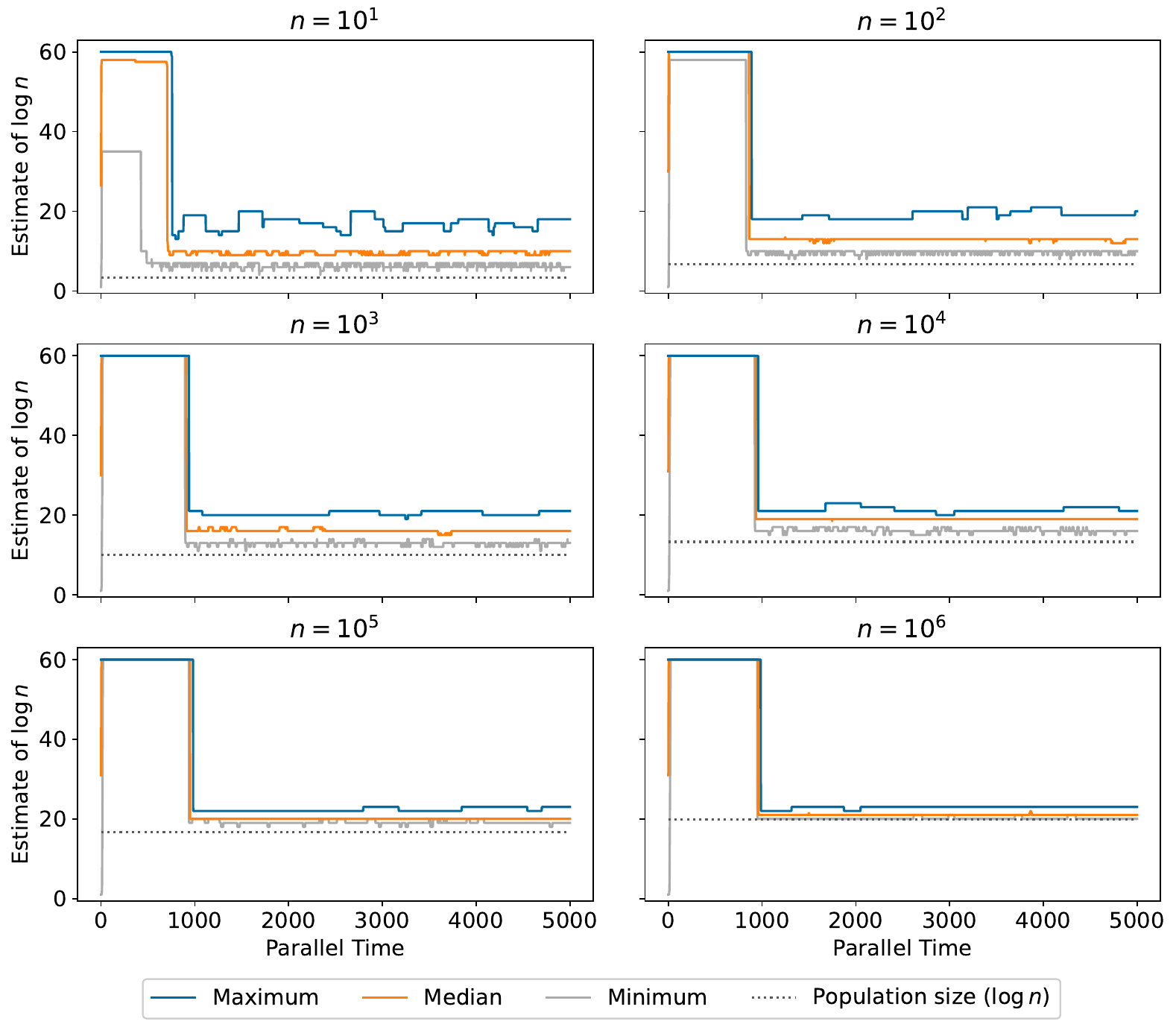}
  \caption[estimates for initial value $M=60$.]{Estimates for different population sizes initialized with an estimate of $60$}
  \Description[A plot showing estimates for initial value $M=60$]{A plot showing estimates for different population sizes initialized with an estimate of $60$}
  \label{fig:adopt_60_results}
\end{figure}

\clearpage

\section{Countdown with Higher Value Propagation} \label{apx:chvp}

In this appendix we present a slightly modified analysis for the detection protocol introduced in \cite{DBLP:conf/dna/AlistarhDKSU17}.
In particular, we extend the analysis to the case of a one-sided variant, where only the responder changes its state.
In the following analysis, we assume the (unbounded) state space
$Q = \{0, 1, 2, \dots\}$
and analyze a protocol based on the following transition rule:
\begin{align}
    \left(x, y\right) &\rightarrow \left(\min\{x,y\}+1, y\right) \label{eq:clvp-2}
\end{align}
Note that this protocol uses an infinite number of states and thus is only useful for our theoretical analysis. We will show later how it captures the behavior of the original protocol.

\restatechvpmax*

\begin{proof}
In our proof we follow along the lines of \cite{DBLP:conf/dna/AlistarhDKSU17}.
Note that this proof thus also uses the opposite of the CHVP process: counting up with lower value propagation (see \cref{eq:clvp-2}).
We assume without loss of generality that initially $l = \min_{v \in V}\{C_0(v)\} = 0$; any $l > 0$ would provide a simple linear offset without affecting the proven time bounds.

Recall that a configuration maps agents to states.
We define the potential function for the configuration $C$ at time $t$ as
\[ \Phi_t(v) = 6^{-C_t(v)}, \Phi_t = \sum_{v \in V} \Phi_t(v) . \]
If both agents $(i,j)$ have the same state, we compute the new potential as
\[ \textbf{E}[\Phi_{t+1} | \Phi_t] = \Phi_t - 6^{-C_t(i)} + 6^{-(C_t(i)+1)} = \Phi_t - \frac{5}{6} 6^{C_t(i)}. \]
Thus, the initiating agent is expected to lose $5/6$ of its potential.
However, if the states differ, we can define $C_t(j) = C_t(i) + d$ with $d \geq 1$ and select the initiator with a coin flip.
Thus, $i$ has a $50\%$ chance of being the initiator and a $50\%$ chance of being the responder.
The expected potential can then be expressed as
\begin{alignat*}{2}
    \textbf{E}[\Phi_{t+1} | \Phi_t] &= \frac{1}{2} \textbf{E}[\Phi_{t+1} | \Phi_t \land \text{$i$ is initiator}] &&+ \frac{1}{2} \textbf{E}[\Phi_{t+1} | \Phi_t \land \text{$i$ is responder}] \\
    &= \frac{1}{2} \left( \Phi_t - 6^{-C_t(i)} + 6^{-(C_t(i) + 1)} \right) &&+ \frac{1}{2} \left( \Phi_t - 6^{-(C_t(i)+d)} + 6^{-(C_t(i)+1)} \right) \\
    &\leq \Phi_t - \frac{1}{3} 6^{-C_t(i)}.
\end{alignat*}
Thus, the initiator's potential is expected to drop by at least $1/3$ in each interaction.
Each agent initiates an interaction with probability $1/n$.
Then the following bound holds:
\begin{align*}
 \textbf{E}[\Phi_{t+1} | \Phi_t] &\leq \Phi_t - \sum_v \text{Pr($v$ initiates in round $t$)} \cdot \frac{1}{3} \Phi_t(v) 
 = \left(1 - \frac{1}{3n}\right) \Phi_t.
 \end{align*}

    Substituting $\Phi_0 \leq n$ and fixing $
    t = 3 \cdot n \ln(6^{\Delta} \cdot n^{k})
    < 7n(\Delta + k \log n)
    $ we find
    \[ \textbf{E}[\Phi_t] \leq \left(1 - \frac{1}{3n}\right)^t \cdot n \leq e^{-\ln(6^{\Delta} \cdot n^{k})} \cdot n = 6^{-\Delta} \cdot \frac{1}{n^{k}}. \]
    By Markov’s inequality, this means that
    \[
\Pr\left[\sum_{v \in V} \Phi_{t}(v)  \geq 6^{-\Delta} \right] \leq n^{-k} ,
    \]
    and thus there is no agent in any of the states $X_0, X_1, \dots, X_{\Delta}$ with probability at least $1 - n^{-k}$.
\end{proof}

In the following lemma, we investigate the upper bound of the CHVP variable.
Initially, the minimum value can be arbitrarily small.
We now show that after $O(\log n)$ time, this minimum will be close to the initial maximum $m$.
To do this, we model the CHVP process as an epidemic starting from $m$.
All agents starting with $m$ start off as infected, with the rest being uninfected.
Once all agents have been infected, their maximum value will be bounded by the number of interactions each agent initiates.

\restatechvpmin*

\begin{proof}
Similar to \cref{lemma:chvp_max}, we use the CLVP process from \cref{eq:clvp-2}, and assume that $l = \min_{v \in V}\{C_0(v)\} = 0$ without loss of generality, as any $l > 0$ represents a linear offset.
In our proof, we couple the CLVP process with a modified process that incorporates epidemic spreading.
This process then provides an upper bound for the maximum value of a CLVP process, and thus a lower bound for a CHVP process.
To this end, we assume that initially at least one agent is ``infected'' and all other agents are ``uninfected''.
When an uninfected agent interacts with an infected agent, it becomes infected as well.

In our process, we define that an uninfected agent has state $-1$, and infected agents have a state larger than or equal to $0$.
Those agents that have value $0$ in the original process CLVP are initially infected in our modified process and thus have value $0$ as well.
Our modified process uses the following transitions at time $t$.
\begin{alignat}{2}
    (u < 0, v < 0) & \rightarrow (u, v) \\
    (u < 0, v \geq 0) & \rightarrow (v + 1, v) \\
    (u\geq 0, v\phantom{{}\geq 0}) & \rightarrow (u + 1, v)
\end{alignat}

In contrast, the transition rule in CLVP reads
\begin{alignat}{2}
    (u, v) & \rightarrow (\min\{u,v\} + 1, v).
\end{alignat}
It is therefore straightforward to verify via a coupling that the value of each infected agent in our modified process is an upper bound on the value of the same agent in CLVP, and thus a lower bound for the original CHVP process.

\def\C#1#2{C_{#1}(#2)}

For every agent $w$ which is initially not infected, we define $\pi(w)$ to be the agent which infected agent $w$ and $\tau(w)$ to be the time when agent $w$ is infected.
For agents that are initially infected, we define $\pi(w) = w$ and $\tau(w) = 0$.
Recall that $\C{t}{u}$ is the random variable that describes the value of agent $u$ at time $t$.
We now consider an arbitrary but fixed agent $v$
and observe that the final value of agent $v$ at time $\tau$ can be written as
\[
\C{\tau}{v} = \C{\tau(\pi(v))}{\pi(v)} + \Bin(\tau - \tau(\pi(v)), 1/n),
\]
where $\Bin(t, 1/n)$ means a random variable with binomial distribution with parameters $t$ and $1/n$.
Indeed, this models the observation that the final value of agent $v$ equals the value it adopts when it becomes infected at time $\tau(v)$ plus the number of increments in subsequent interaction.
These increments by $1$ occur when agent $v$ initiates an interaction, that is, with probability $1/n$.
Recursively applying this observation gives us
\begin{align*}
\C{\tau}{v} &= \C{\tau(\pi(\pi(v)))}{\pi(\pi(v))} + \Bin(\tau(\pi(v)) - \tau(\pi(\pi(v))),1/n) + \Bin(\tau - \tau(\pi(v)),1/n) \\
&= \C{\tau(\pi(\pi(v)))}{\pi(\pi(v))} + \Bin(\tau - \tau(\pi(\pi(v)),1/n)\\
& \dots
\intertext{and ultimately, using $C_0(w) = 0$ for any agent $w$ that is initially infected, we get}
\C{\tau}{v} &= \Bin(\tau,1/n).
\end{align*}
We apply Chernoff bounds to $\C{\tau}{v}$ and get
\begin{align*}
\Pr[{\C{\tau}{v} \geq (1 + \delta) E[\C{\tau}{v}] }] &= \Pr\left[{\C{\tau}{v} \geq \left(1 + \sqrt{\frac{3 \ln 2}{7}\frac{k+1}{k}}\right) 7(\Delta + k \log n) }\right] \\
    &\leq e^{-\frac{k+1 }{k} \ln 2} n^{-1-k} \\
    &\leq n^{-1-k}.
\end{align*}

The lemma now follows from the union bound over all agents together with the observation from \cref{lemma:agent_participations} that at time $\tau$ all agents are infected w.h.p.
\end{proof}

\end{document}